\documentclass[runningheads]{llncs}
\pdfoutput=1
\usepackage[crop]{llncsconf} 
\usepackage{graphicx}
\usepackage{amsmath,amssymb,mathtools}
\usepackage{cleveref}
\usepackage{cite}

\usepackage[table,xcdraw]{xcolor} 

\usepackage[mathlines]{lineno}
\let\doendproof\endproof
\renewcommand\endproof{~\hfill$\qed$\doendproof}

\newcommand{\onestep}{\leftrightarrow}

\newcommand{\TS}{\mathsf{TS}}
\newcommand{\TJ}{\mathsf{TJ}}

\newcommand{\setU}{U}
\newcommand{\Uini}{\setU^\mathsf{s}}
\newcommand{\Utar}{\setU^\mathsf{t}}

\newcommand{\RISR}[1]{\mathrm{RISR}_{#1}}
\newcommand{\CRISR}[1]{\mathrm{CRISR}_{#1}}

\newenvironment{listing}[1]{%
        \begin{list}{*}{%
                 \settowidth{\labelwidth}{#1}%
                 \setlength{\leftmargin}{\labelwidth}%
                 \advance \leftmargin by 12pt
                   \setlength{\itemsep}{0pt}%
                   \setlength{\parsep}{0pt}%
                   \setlength{\topsep}{0pt}%
                   \setlength{\parskip}{0pt}%
}%
}{%
\end{list}}

\newcounter{one}
\newcounter{two}
\newcounter{three}
\newcounter{four}
\newcounter{five}
\definecolor{darkred}{rgb}{0.7,0,0}

\newtheorem{observation}[theorem]{Observation}

\usepackage[appendix=inline]{apxproof} 
\newtheoremrep{theorem}{Theorem}
\newtheoremrep{lemma}[theorem]{Lemma}
\newtheoremrep{observation}[theorem]{Observation}
\newtheoremrep{corollary}[theorem]{Corollary}

\usepackage{ifthen}
\usepackage{tikz}
\usetikzlibrary{calc}
\usetikzlibrary{decorations.pathreplacing}
\usetikzlibrary{arrows.meta}
\usetikzlibrary{bending}

\def\yScale{0.6}
\def\xScale{2}
\definecolor{mikadoyellow}{rgb}{1.0, 0.77, 0.05}
\definecolor{limegreen}{rgb}{0.2, 0.8, 0.2}
\def\iniClr{mikadoyellow}
\def\tarClr{limegreen}
\tikzset{
    side by side/.style 2 args={
            draw=#1,
            line width=4pt,
            postaction={
                    clip, postaction={draw=#2}
                }
        }
}
\newcommand{\convexpath}[2]{
    [
            create hullnodes/.code={
                    \global\edef\namelist{#1}
                    \foreach [count=\counter] \nodename in \namelist {
                        \global\edef\numberofnodes{\counter}
                        \node at (\nodename) [draw=none,name=hullnode\counter] {};
                    }
                    \node at (hullnode\numberofnodes) [name=hullnode0,draw=none] {};
                    \pgfmathtruncatemacro\lastnumber{\numberofnodes+1}
                    \node at (hullnode1) [name=hullnode\lastnumber,draw=none] {};
                },
            create hullnodes
        ]
    ($(hullnode1)!#2!-90:(hullnode0)$)
    \foreach [
        evaluate=\currentnode as \previousnode using \currentnode-1,
        evaluate=\currentnode as \nextnode using \currentnode+1
    ] \currentnode in {1,...,\numberofnodes} {
            -- ($(hullnode\currentnode)!#2!-90:(hullnode\previousnode)$)
            let \p1 = ($(hullnode\currentnode)!#2!-90:(hullnode\previousnode) - (hullnode\currentnode)$),
            \n1 = {atan2(\y1,\x1)},
            \p2 = ($(hullnode\currentnode)!#2!90:(hullnode\nextnode) - (hullnode\currentnode)$),
            \n2 = {atan2(\y2,\x2)},
            \n{delta} = {-Mod(\n1-\n2,360)}
            in
                {arc [start angle=\n1, delta angle=\n{delta}, radius=#2]}
        }
    -- cycle
}
\begin{document}
\title{Reconfiguration of Regular Induced Subgraphs\thanks{%
Partially supported by JSPS KAKENHI Grant Numbers 
  JP18H04091, 
  JP18K11168, 
  JP18K11169, 
  JP19K11814, 
  JP19K20350, 
  JP20H05793, 
  JP20K19742, 
  JP21H03499, 
  JP21K11752. 
}
}


\author{
Hiroshi Eto\inst{1}
\and
Takehiro Ito\inst{1}
\and
Yasuaki Kobayashi\inst{2}
\and
Yota Otachi\inst{3}
\and
Kunihiro~Wasa\inst{4}
}

\authorrunning{H. Eto et al.}

\institute{%
%
%
Graduate School of Information Sciences, Tohoku University, Sendai, Japan\\
\email{\{hiroshi.eto.b4,takehiro\}@tohoku.ac.jp}
\and
Graduate School of Informatics, Kyoto University, Kyoto, Japan\\
\email{kobayashi@iip.ist.i.kyoto-u.ac.jp}
\and
Nagoya University, Nagoya, Japan\\
\email{otachi@nagoya-u.jp}
\and
Toyohashi University of Technology, Toyohashi, Japan\\
\email{wasa@cs.tut.ac.jp}
}

\maketitle

\begin{abstract}
We study the problem of checking the existence of a step-by-step transformation of $d$-regular induced subgraphs in a graph, where $d \ge 0$ and each step in the transformation must follow a fixed reconfiguration rule. 
Our problem for $d=0$ is equivalent to \textsc{Independent Set Reconfiguration}, which is one of the most well-studied reconfiguration problems.
In this paper, we systematically investigate the complexity of the problem, in particular, on chordal graphs and bipartite graphs. 
Our results give interesting contrasts to known ones for \textsc{Independent Set Reconfiguration}. 
\end{abstract}

\section{Introduction}
Combinatorial reconfiguration~\cite{IDHPSUU11,vandenHeuvel13,Nishimura18} studies the reachability in the solution space formed by feasible solutions of an instance of a search problem. 
In a reconfiguration problem, we are given two feasible solutions of a search problem and are asked 
to determine whether we can modify one to the other by repeatedly applying a prescribed reconfiguration rule while keeping the feasibility.
Such problems arise in many applications and
studying them is important also for understanding the underlying problems deeper
~(see the surveys \cite{vandenHeuvel13,Nishimura18} and the references therein).

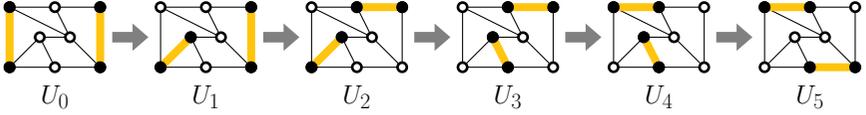
\begin{figure}[htb]
\centering
    \newcommand{\exampleBaseGraph}{
    \node at (0, 0) (v1) {};
    \node at (1.5, 0) (v2) {};
    \node at (3, 0) (v3) {};
    \node at (1, 1) (v4) {};
    \node at (2, 1) (v5) {};
    \node at (0, 2) (v6) {};
    \node at (1.5, 2) (v7) {};
    \node at (3, 2) (v8) {};
    \foreach \u/\v in {
            v1/v2, v1/v4, v1/v6,
            v2/v3, v2/v4,
            v3/v5, v3/v8,
            v4/v5,
            v5/v6, v5/v7,
            v6/v7,
            v7/v8} {
            \draw (\u) edge (\v);
        }

}

\def\exscale{0.4}
\begin{tikzpicture}[scale = \exscale,
        every node/.style = {scale=\exscale, circle, line width=1pt, draw=black, fill=white},
        token/.style = {draw=black, fill=black},
        matching edge/.style = {draw=\iniClr, line width=3pt},
        moved matching edge/.style = {draw=\iniClr, line width=2pt, dashed},
        tar matching edge/.style = {draw=\tarClr, line width=3pt}
    ]
    \def\shiftAmount{5}
    \foreach \i/\meu/\mev/\mfu\mfv in {
            0/v1/v6/v3/v8,
            1/v3/v8/v1/v4,
            2/v1/v4/v7/v8,
            3/v7/v8/v2/v4,
            4/v2/v4/v6/v7,
            5/v6/v7/v2/v3%
        }{
            \begin{scope}[shift={(\i*\shiftAmount, 0)}]
                \exampleBaseGraph
                \ifthenelse{\equal{\i}{0}}{}{
                \draw[-{Triangle[width=10pt,length=6pt]}, line width=4pt, gray] ($(v1)!0.5!(v6) - (1.6, 0)$) to ($(v1)!0.5!(v6) - (0.4, 0)$);
                }
                \draw (\meu) edge[matching edge] (\mev);
                \draw (\mfu) edge[matching edge] (\mfv);
                \foreach \u in {\meu, \mev, \mfu, \mfv} {
                        \node[token] at (\u) {};
                    }
                \Huge
                \node[draw=none, fill=none] at ($(v2) + (0, -1)$) {$\setU_\i$};
            \end{scope}
        }

\end{tikzpicture}
    \caption{Example of a $\TJ$-sequence for $d=1$, where each vertex in a $d$-regular set $\setU_i$ is colored with black. Under $\TS$, only $\setU_2 \onestep \setU_3$ holds.}
    \label{fig:ex}
\end{figure}

\subsection{Our problems}
In this paper, we take $d$-regular induced subgraphs in a graph as feasible solutions of the solution space.
Recall that a graph is \emph{$d$-regular} if every vertex in the graph is of degree $d$.
By the choice of $d$, we can represent some well-known graph properties: 
$d=0$ corresponds to an independent set of a graph $G$, and 
$d=1$ corresponds to an induced matching of $G$. 
If we require $d$-regular induced subgraphs to be connected, then $d=2$ corresponds to an induced cycle of $G$.  

We then define two reconfiguration rules on the $d$-regular induced subgraphs. 
Since we consider only induced subgraphs of a graph $G$, each feasible solution can be represented by a vertex subset $\setU$ of $G$. 
We denote by $G[\setU]$ the subgraph of $G$ induced by $\setU$. 
We say that a vertex subset $\setU$ of a graph $G$ is a \emph{$d$-regular set of $G$} if $G[\setU]$ is $d$-regular. 
Then, there are two well-studied reconfiguration rules~\cite{KaminskiMM12}, called Token Jumping rule ($\TJ$ rule for short) and Token Sliding rule ($\TS$ rule for short).\footnote{
There is another well-studied rule, Token Addition and Removal ($\mathsf{TAR}$)~\cite{KaminskiMM12}.
We are not going to consider this rule as it cannot keep $d$-regularity unless $d=0$.}  
Let $\setU$ and $\setU'$ be two $d$-regular sets of $G$.
Then, we write
\begin{listing}{a}
  \item[$\bullet$] \emph{$\setU \onestep \setU'$ under $\TJ$} if $|\setU \setminus \setU'| = |\setU' \setminus \setU| = 1$; and
  \smallskip

  \item[$\bullet$] \emph{$\setU \onestep \setU'$ under $\TS$} if $\setU \setminus \setU' = \{v\}$, $\setU' \setminus \setU = \{w\}$, and $vw \in E(G)$.
  \smallskip
\end{listing}
A sequence $\langle \setU_{0}, \setU_{1}, \ldots, \setU_{\ell} \rangle$ of $d$-regular sets of $G$ is called a \emph{reconfiguration sequence under} $\TJ$ (or $\TS$) between two $d$-regular sets $\setU_{0}$ and $\setU_{\ell}$ if $\setU_{i-1} \onestep \setU_{i}$ holds under $\TJ$ (resp., $\TS$) for all $i \in \{1, 2, \ldots, \ell\}$.
A reconfiguration sequence under $\TJ$ (or $\TS$) is simply called a \emph{$\TJ$-sequence} (resp., \emph{$\TS$-sequence}).
Note that all $d$-regular sets in the sequence have the same cardinality.

We now define the problem for a rule $\mathsf{R} \in \{\TJ, \TS\}$, as follows:
	\begin{center}
	\fbox{
		\parbox{0.9\hsize}{
            \textsc{$d$-Regular Induced Subgraph Reconfiguration} under $\mathsf{R}$ (abbreviated as $\RISR{d}$)
            \begin{listing}{\textbf{Question:}}
            \item[\textbf{Input:}] A graph $G$ and $d$-regular sets $\Uini$ and $\Utar$ of $G$.
            \item[\textbf{Question:}] Is there an $\mathsf{R}$-sequence between $\Uini$ and $\Utar$?
            \end{listing}
		}
	}
	\end{center}
A $d$-regular set $\setU$ of $G$ is \emph{connected} if $G[\setU]$ is connected. 
We also consider the following special case of $\RISR{d}$, which only allows connected sets as the initial and target $d$-regular sets:
	\begin{center}
	\fbox{
		\parbox{0.9\hsize}{
            \textsc{Connected $d$-Regular Induced Subgraph Reconfiguration} under $\mathsf{R}$ (abbreviated as $\CRISR{d}$)
            \begin{listing}{\textbf{Question:}}
            \item[\textbf{Input:}] A graph $G$ and connected $d$-regular sets $\Uini$ and $\Utar$ of $G$.
            \item[\textbf{Question:}] Is there an $\mathsf{R}$-sequence between $\Uini$ and $\Utar$?
            \end{listing}
		}
	}
	\end{center}
Although $\CRISR{d}$ does not explicitly ask intermediate $d$-regular sets to be connected,
it is actually forced by the $d$-regularity and connectivity of the initial set (see Section~\ref{sec:pre}).


\subsection{Known and related results}
Hanaka et al.~\cite{HanakaIMMNSSV20} introduced \textsc{Subgraph Reconfiguration}, which unifies several reconfiguration problems where feasible solutions are defined as
(induced or non-induced)
subgraphs in a graph satisfying a specific property. 
They considered several graph properties for defining feasible solutions, and one of their results shows that $\CRISR{2}$ is PSPACE-complete under $\TS$ and $\TJ$. 

As related work, Ito et al.~\cite{ItoOO15} introduced \textsc{Clique Reconfiguration}, which can be seen as a special case of $d$-regular induced subgraphs. 
The problem is solvable in polynomial time on even-hole-free graphs (and hence on chordal graphs) under $\TJ$ and $\TS$~\cite{ItoOO15}.
M{\"{u}}hlenthaler~\cite{Muhlenthaler15} proved that \textsc{Subgraph Reconfiguration} is solvable in polynomial time when (not necessarily induced) regular graphs are taken as feasible solutions. His result can be generalized to degree-constrained subgraphs where each vertex has lower and upper bounds for its degree.

\textsc{Independent Set Reconfiguration}, equivalent to $\RISR{0}$, is one of the most well-studied reconfiguration problems. 
$\RISR{0}$ under $\TJ$ is PSPACE-complete on perfect graphs~\cite{KaminskiMM12}, and is NP-complete on bipartite graphs~\cite{LokshtanovM19}, whereas it is solvable in polynomial time on even-hole-free graphs~\cite{KaminskiMM12}, claw-free graphs~\cite{BonsmaKW14}, and cographs~\cite{Bonsma16}.
On the other hand, $\RISR{0}$ under $\TS$ is PSPACE-complete on bipartite graphs~\cite{LokshtanovM19}, and on split graphs~\cite{BelmonteKLMOS21}. 

These precise complexity analyses of $\RISR{0}$ show interesting contrast with respect to the reconfiguration rules $\TS$ and $\TJ$. (See Table~\ref{tbl:summary}.)
On chordal graphs, tractability of $\RISR{0}$ depends on the choice of reconfiguration rules. 
On bipartite graphs, the complexity of $\RISR{0}$  shows arguably the most surprising behavior depending on the reconfiguration rules.
Lokshtanov and Mouawad~\cite{LokshtanovM19}
showed that, on bipartite graphs, 
$\RISR{0}$ is PSPACE-complete under $\TS$ 
but NP-complete under $\TJ$~\cite{LokshtanovM19}.
That is, $\RISR{0}$ is intractable under both rules but in different senses.

\begin{table}[t]
\centering
\caption{Summary of the results.}
\label{tbl:summary}
\begin{tabular}{c|l|l|l|l}
  & \multicolumn{2}{l|}{\cellcolor{lightgray!50} $\RISR{d}$}
  & \multicolumn{2}{l}{\cellcolor{lightgray!50} $\CRISR{d}$ ($d \ge 2$)}
  \\
  & \cellcolor{lightgray!25} $\TS$
  & \cellcolor{lightgray!25} $\TJ$
  & \cellcolor{lightgray!25} $\TS$
  & \cellcolor{lightgray!25} $\TJ$
  \\
  \hline
  \cellcolor{lightgray!50} \begin{tabular}[c]{@{}c@{}}constant\\ bandwidth\end{tabular}
  & \multicolumn{2}{c|}{PSPACE-c [Cor~\ref{cor:bandwidth_R}]}
  & \multicolumn{2}{c}{
         \begin{tabular}[c]{@{}c@{}}
        PSPACE-c \\
        $[$Thm~\ref{thm:CRISR-general}$]$
        \end{tabular}}
  \\ 
  \hline
  \cellcolor{lightgray!50} chordal
  &
  \begin{tabular}[c]{@{}l@{}}
    $d=0$: PSPACE-c \cite{BelmonteKLMOS21} \\
    $d \ge 1$: PSPACE-c \\
    ~~~~~~~~~~~~~~~[Thm~\ref{thm:RISR-chordal}]
  \end{tabular}   &
  \begin{tabular}[c]{@{}l@{}}
    $d=0$: P \cite{KaminskiMM12} \\
    $d \ge 1$: PSPACE-c \\
    ~~~~~~~~~~~~[Thm~\ref{thm:RISR-chordal}]
  \end{tabular} 
  & \multicolumn{2}{c}{P \cite{ItoOO15}}
  \\
  \hline
  \cellcolor{lightgray!50} bipartite 
  &
  \begin{tabular}[c]{@{}l@{}}
    $d=0$: PSPACE-c \cite{LokshtanovM19} \\
    $d \ge 1$: P [Obs~\ref{obs:RISR-bigraph-TS}]
  \end{tabular}
  &
  \begin{tabular}[c]{@{}l@{}}
    $d=0$: NP-c \cite{LokshtanovM19} \\
    $d \ge 1$: PSPACE-c \\
    ~~~~~~~~~~~~[Thm \ref{them:RISR-bigraph-TJ}]
  \end{tabular}
  & P [Obs~\ref{obs:RISR-bigraph-TS}]
  & 
  \begin{tabular}[c]{@{}c@{}}
    PSPACE-c \\
    $[$Thm \ref{thm:CRISR-bigraph-TJ}$]$
  \end{tabular}
\end{tabular}
\end{table}

\subsection{Our contribution}
In this paper, we investigate the complexity of $\RISR{d}$ and $\CRISR{d}$ systematically. 
In particular, we focus on chordal graphs and bipartite graphs, where the complexity of $\RISR{0}$ (i.e., \textsc{Independent Set Reconfiguration}) is known to show interesting behavior depending on the reconfiguration rules. 
Our results are summarized in Table~\ref{tbl:summary} together with known results for $\RISR{0}$.

We note that our results for $d \ge 1$ give interesting contrasts to known ones for $d=0$.
On chordal graphs, both $\TS$ and $\TJ$ rules have the same complexity for $d \ge 1$, whereas they are different for $d=0$.
%
On bipartite graphs, the complexity of $\RISR{d}$ for $d \ge 1$ shows a kind of reverse phenomenon from $d=0$: the problem becomes easier (indeed, becomes a trivial case as in Observation~\ref{obs:RISR-bigraph-TS}) under $\TS$, but becomes harder under $\TJ$.

\section{Preliminaries}
\label{sec:pre}

In this paper, we only consider simple and undirected graphs.
Let $G = (V,E)$ be a graph.
We sometimes denote by $V(G)$ and $E(G)$ the vertex and edge sets of $G$, respectively. 
For a vertex subset $\setU$ of $G$, we denote by $G[\setU]$ the subgraph of $G$ induced by $\setU$.
We say that $\setU$ is \emph{connected} if $G[\setU]$ is connected. 

In the introduction, we have defined the problems $\RISR{d}$ and $\CRISR{d}$ under $\mathsf{R}$ for a rule $\mathsf{R} \in \{\TJ, \TS\}$. 
We denote by $\langle G, \Uini, \Utar \rangle$ an instance of these problems. 
Since the problems clearly belong to PSPACE, we will only show PSPACE-hardness for proving PSPACE-completeness.

Note that although $\CRISR{d}$ only asks the input sets $\Uini$ and $\Utar$ to be connected,
it is actually required that all sets in a reconfiguration sequence are connected
by Lemma~\ref{lem:isomorphism} below.
\begin{lemmarep}
\label{lem:isomorphism}
Let $\setU$ and $\setU'$ be $d$-regular sets of a graph $G$.
If $\setU \onestep \setU'$ under $\TJ$ or $\TS$,
then $G[\setU]$ and $G[\setU']$ are isomorphic.
\end{lemmarep}
\begin{proof}
Since the adjacency under $\TS$-sequence is a stronger assumption than the one under $\TJ$-sequence,
it suffices to prove the statement under $\TJ$.
If $d = 0$, then the statement is trivial as both $\setU$ and $\setU'$
independent sets of the same size.
In the following, we assume that $d \ge 1$.

Let $u, v$ be vertices such that $\setU \setminus \{u\} = \setU' \setminus \{v\}$.
Since $\setU$ and $\setU'$ are $d$-regular and $\setU \setminus \{u\} = \setU' \setminus \{v\}$, we have $\setU \setminus N[u] = \setU' \setminus N[v]$ and 
$\setU \cap N(u) = \setU' \cap N(v)$.
Thus, a function $f: \setU \to \setU'$ defined by setting $f(w) = w$ for $w \in \setU \setminus \{u\}$ and $f(u) = v$ is an isomorphism between $\setU$ and $\setU'$. 
\end{proof}

For a positive integer $k$, we define $[k] = \{d \in \mathbb{Z} \mid 1 \le d \le k\}$.
For an $n$-vertex graph $G = (V,E)$, the \emph{width} of a bijection $\pi \colon V \to [n]$ is $\max_{\{u,v\} \in E} |\pi(u) - \pi(v)|$.
The \emph{bandwidth} of $G$ is the minimum width over all bijections $\pi \colon V \to [n]$.
The graphs of constant bandwidth are quite restricted in the sense that
the bandwidth of a graph is an upper bound of pathwidth (and of treewidth)~\cite{SorgeW19}.
For example, a path has bandwidth $1$ and a cycle has bandwidth $2$.

For a positive integer $t$, a graph $H$ is a \emph{$t$-sketch} of a graph $G$
if there exists a mapping $f \colon V(G) \to V(H)$ such that 
$|f^{-1}(v)| \le t$ for every $v \in V(H)$ and
$\{u,v\} \in E(G)$ implies either $f(u) = f(v)$ or $\{f(u), f(v)\} \in E(H)$.
The following lemma is useful for bounding the bandwidth of a graph.
\begin{lemmarep}
\label{lem:bw}
Let $H$ be a graph of bandwidth at most $b$.
If $H$ is a $t$-sketch of $G$, then the bandwidth of $G$ is at most $t(b+1)$.
\end{lemmarep}
\begin{proof}
Let $f \colon V(G) \to V(H)$ be a mapping in the definition of $t$-sketch.
Let $\pi_{H} \colon V(H) \to [|V(H)|]$ be a bijection of width at most $b$.
Let $\pi_{G} \colon V(G) \to [|V(G)|]$ be a bijection such that for all distinct $u,v \in V(G)$,
$\pi_{G}(u) < \pi_{G}(v)$ implies $\pi_{H}(f(u)) \le \pi_{H}(f(v))$.
We show that $\pi_{G}$ has width at most $t(b+1)$.

Let $\{u,v\} \in E(G)$. Assume by symmetry that $\pi_{G}(u) < \pi_{G}(v)$.
If $f(u) = f(v)$, then $f(w) = f(u)$ ($=f(v)$) holds for all $w \in V(G)$ with $\pi_{G}(u) \le \pi_{G}(w) \le \pi_{G}(v)$,
and thus $\pi_{G}(v) - \pi_{G}(u) < t$.
Next assume that $\pi_{H}(f(u)) < \pi_{H}(f(v))$.
Then, $\pi_{H}(f(u)) \le \pi_{H}(f(w)) \le \pi_{H}(f(v))$ holds for all $w \in V(G)$ with $\pi_{G}(u) \le \pi_{G}(w) \le \pi_{G}(v)$.
Since $\pi_{H}(f(v)) - \pi_{H}(f(u)) \le b$
and $|f^{-1}(v)| \le t$ for every $v \in V(H)$, it holds that
$\pi_{G}(v) - \pi_{G}(u) < t(b+1)$.
\end{proof}

\section{Complexity of $\RISR{}$ and $\CRISR{}$}

In this section, we show that $\RISR{}$ and $\CRISR{}$ are intractable 
on the classes of chordal graphs, bipartite graphs, and graphs of bounded bandwidth.
We also observe that some cases are tractable.

Observe that a connected $0$-regular graph is a single vertex
and a connected $1$-regular graph is a single edge.
This implies that, for $d \le 1$, $\CRISR{d}$ is polynomial-time solvable under $\TJ$ and $\TS$.
Hence, we only consider the case of $d \ge 2$ for $\CRISR{d}$.
On the other hand, we cannot put such an assumption for the more general problem $\RISR{d}$
as $\RISR{0}$ is PSPACE-complete under $\TJ$ and $\TS$~\cite{KaminskiMM12}.

\subsection{Chordal graphs}

For chordal graphs, let us  first consider $\CRISR{d}$, which is actually easy.
Since every connected regular induced subgraph in a chordal graph is a complete graph~\cite{AsahiroEIM14},
$\CRISR{d}$ on chordal graphs can be solved by using 
a linear-time algorithm for \textsc{Clique Reconfiguration} on chordal graphs~\cite{ItoOO15}.
Also, we can use the same algorithm for the general $\RISR{d}$ with $d \ge 1$ on split graphs
since a split graph can have at most one nontrivial connected component.
\begin{observation}
\label{obs:RISR-split}
$\CRISR{d}$ with $d \ge 0$ on chordal graphs 
and $\RISR{d}$ with $d \ge 1$ on split graphs 
can be solved in linear time.
\end{observation}

Now let us turn our attention to the general problem $\RISR{d}$ on chordal graphs.
When $d=0$, the complexity depends on the reconfiguration rules.
On chordal graphs, $\RISR{0}$ is PSPACE-complete under $\TS$~\cite{BelmonteKLMOS21},
while it is trivially polynomial-time solvable under $\TJ$ as it does not have a no-instance~\cite{KaminskiMM12}.
For $d \ge 1$, we show that the problem is PSPACE-complete under both rules.
\begin{theorem}
\label{thm:RISR-chordal}
For every constant $d \geq 1$,
$\RISR{d}$ is PSPACE-complete on chordal graphs under $\TJ$ and $\TS$.
\end{theorem}
\begin{proof}
We give a polynomial-time reduction from \textsc{Independent Set Reconfiguration} on chordal graphs under $\TS$.
Let $\langle H, I, I' \rangle$ be an instance of \textsc{Independent Set Reconfiguration} under $\TS$ on chordal graphs.
From this instance, we construct an instance $\langle G, \Uini, \Utar \rangle$ of $\RISR{d}$ under $\TJ$ and $\TS$ on chordal graphs as follows.
For each $v \in V(H)$, we take a set $X_{v}$ of $d+1$ vertices.
We set $V(G) = \bigcup_{v \in V(H)} X_{v}$.
Each $X_{v}$ is a clique in $G$.
We add all possible edges between $X_{u}$ and $X_{v}$ if $\{u,v\} \in E(H)$.
We set $\Uini = \bigcup_{v \in I} X_{v}$ and $\Utar = \bigcup_{v \in I'} X_{v}$.
This reduction can be seen as repeated additions of true twin vertices,\footnote{%
Two vertices $u,v$ are \emph{true twins} if $N[u] = N[v]$.}
which do not break the chordality, and thus $G$ is chordal.

We now show that
$\langle H, I, I' \rangle$ is a yes-instance of \textsc{Independent Set Reconfiguration}
under $\TS$
if and only if
$\langle G, \Uini ,\Utar \rangle$ is a yes-instance of $\RISR{d}$ under $\TJ$ and $\TS$.

To prove the only-if direction, we assume that there exists a $\TS$-sequence $\langle 
I_0, \dots, I_\ell \rangle$ between $I$ and $I'$. 
For each $i$ with $0 \leq i \leq \ell$, 
we set $R_{i} = \bigcup_{v \in I_i} X_v$.
Observe that each $R_{i}$ is a $d$-regular set, $R_0 = \Uini$, and $R_{\ell} =\Utar$.
Hence, it suffices to show that there is a $\TS$-sequence (and thus a $\TJ$-sequence as well) from $R_{i}$ to $R_{i+1}$ for all $1 \le i \le \ell-1$.
Let $I_{i} \setminus I_{i+1} = \{p\}$ and $I_{i+1} \setminus I_{i} = \{q\}$,
and thus $R_{i} \setminus R_{i+1} = X_{p}$ and $R_{i+1} \setminus R_{i} = X_{q}$.
Let $X_{p} = \{p_{h} \mid 1 \le h \le {d+1}\}$ and $X_{q} = \{q_{h} \mid 1 \le h \le {d+1}\}$.
We set $T^{0} = R_{i}$, and for each $j$ with $1 \le j \le d+1$, we define
$T^{j} = T^{j-1} \setminus \{p_j\} \cup \{q_j\}$.
Observe that each $T^{j}$ is $d$-regular and $T^{d+1}=R_{i+1}$.
We can see that the exchanged vertices $p_{j}$ and $q_{j}$ are adjacent in $G$
since $p$ and $q$ are adjacent in $H$.
This implies that $\langle T^{0}, \dots, T^{d+1}\rangle$
is a $\TS$-sequence between $R_{i}$ and $R_{i+1}$.


To show the if direction, 
we first show that the existence of a $\TJ$-sequence between $\Uini$ and $\Utar$
implies the existences of a $\TS$-sequence between them as well.
This allows us to start the proof of this direction
with the stronger assumption of having a $\TS$-sequence.

Let $\langle R_{0}, \dots, R_{\ell} \rangle$ be a $\TJ$-sequence between $\Uini$ and $\Utar$.
For each $i \in [\ell]$, 
we show that there is a $\TS$-sequence between $R_{i-1}$ and $R_{i}$.
Let $\{u\} = R_{i-1} \setminus R_{i}$ and $\{v\} = R_{i} \setminus R_{i-1}$.
If $\{u,v\} \in E(G)$, then $R_{i-1} \onestep R_{i}$ under $\TS$.
Assume that $\{u,v\} \notin E(G)$.
Let $w \in R_{i-1} \cap R_{i}$ be a common neighbor of $u$ and $v$.
Such a vertex exists since $d \ge 1$.
Let $z$ be the vertex of $H$ such that $w \in X_{z}$.
Observe that $u,v \notin X_{z}$ as $\{u,v\} \notin E(G)$.
Now $X_{z} \not\subseteq R_{i-1} \cap R_{i}$ ($= R_{i-1} \setminus \{u\}$) holds since otherwise
$u$ has at least $|X_{z}|=d+1$ neighbors in $G[R_{i-1}]$ as $X_{z}$ is a set of true twins.
Since $u, v \notin X_{z}$, it holds that $X_{z} \not\subseteq R_{i-1} \cup R_{i}$.
Thus there exists $w' \in X_{z}$ that does not belong to $R_{i-1} \cup R_{i}$.
We claim that 
$R_{i-1} \onestep (R_{i-1} \setminus \{u\} \cup \{w'\}) = (R_{i-1} \setminus \{v\} \cup \{w'\}) \onestep R_{i}$ under $\TS$.
We only need to show that the intermediate set $R_{i-1} \setminus \{u\} \cup \{w'\}$ is $d$-regular.
By Lemma~\ref{lem:isomorphism} and the definition of $R_{i-1}$,
each connected component in $G[R_{i-1}]$ is a $(d+1)$-clique.
Let $C$ be the $(d+1)$-clique $G[R_{i-1}]$ that includes $u$ and $w$.
Since $w$ and $w'$ are true twins, the $(d+1)$-clique $C \setminus \{u\} \cup \{w'\}$ is a connected component of $G[R_{i-1} \setminus \{u\} \cup \{w'\}]$.
Thus, $R_{i-1} \setminus \{u\} \cup \{w'\}$ is $d$-regular.

Now assume that there is 
a $\TS$-sequence $\langle R_0, R_1, \dots, R_{\ell} \rangle$ between $\Uini$ and $\Utar$.
From this sequence, we construct a sequence 
$\langle S_{0}, S_{1}, \dots, S_{\ell}\rangle$ of independent sets in $G$ as follows.
Recall that $R_{0} = \bigcup_{v \in I_{0}} X_{v}$.
To construct $S_{0}$, we pick one arbitrary vertex from $X_{v}$ for each $v \in I_{0}$.
That is, $S_{0}$ is a set such that $|S_{0}| = |I|$
and $|S_{0} \cap X_{v}| = 1$ for each $v \in I_{0}$.
For $0 \leq i \leq \ell-1$, we define
\[
S_{i+1} =
\begin{cases}
S_{i} & S_{i} \subseteq R_{i+1}, \\
S_{i} \setminus \{u\} \cup \{v\} & R_{i} \setminus R_{i+1} = \{u\} \subseteq S_{i}, \ 
R_{i+1} \setminus R_{i} = \{v\}.
\end{cases}
\]
Intuitively, $S_{0} \subseteq R_{0}$ can be seen as the set of representatives of all $X_{v}$ contained in $R_{0}$.
For $i \ge 1$, the set $S_i \subseteq R_{i}$ traces the tokens started on the representatives chosen to $S_{0}$.
Each $S_i$ is an independent set of $G$ with size $|I|$.
For $0 \le i \le \ell$,
we now construct $I_{i}$ by projecting the vertices in $S_{i}$ onto $V(H)$;
that is, we define
\[
I_{i} = \{v \in V(H) \mid X_{v} \cap S_{i} \ne \emptyset\}.
\]
Since each $S_{i}$ is an independent set, each $I_{i}$ is an independent set too.
Assume that $I_{i} \ne I_{i+1}$ for some $i$,
$I_{i} \setminus I_{i+1} = \{u\}$, and $I_{i+1} \setminus I_{i} = \{v\}$.
Then, $S_{i} \ne S_{i+1}$ holds.
In particular, 
$S_{i} \setminus S_{i+1} \subseteq X_{u}$ and $S_{i+1} \setminus S_{i} \subseteq X_{v}$.
By the definition of the sequence $\langle S_{0}, \dots, S_{\ell} \rangle$,
this further implies that 
$R_{i} \ne R_{i+1}$,
$R_{i} \setminus R_{i+1} \subseteq X_{u}$,
and $R_{i+1} \setminus R_{i} \subseteq X_{v}$.
Since $\langle R_{0}, \dots, R_{\ell} \rangle$ is a $\TS$-sequence in $G$,
the vertices $u$ and $v$ are adjacent in $H$ by the definition of $G$.
Hence, $I_{i} \onestep I_{i+1}$ under $\TS$.
Therefore, by skipping the sets $I_{i}$ with $I_{i} = I_{i+1}$ in
$\langle I_{0}, \dots, I_{\ell} \rangle$,
we have a $\TS$-sequence from $I$ to $I'$.
\end{proof}

We can use the same reduction in the proof of Theorem~\ref{thm:RISR-chordal} from the same problem but on a different graph class
to show the PSPACE-completeness on graphs of bounded bandwidth.
\begin{corollary}\label{cor:bandwidth_R}
For every constant $d \ge 0$,
there is a constant $b_{d}$ depending only on $d$
such that $\RISR{d}$ is PSPACE-complete under $\TJ$ and $\TS$
on graphs of bandwidth at most $b_{d}$.
\end{corollary}
\begin{proof}
Wrochna~\cite{Wrochna18} showed that there exists a constant $b_{0}$ such that $\RISR{0}$
is PSPACE-complete under $\TJ$ and $\TS$ on graphs of bandwidth at most $b_{0}$.
Let $H$ be a graph of bandwidth at most $b_{0}$.
By applying the reduction in the proof of Theorem~\ref{thm:RISR-chordal} to $H$,
we obtain a graph $G$ of bandwidth at most $b_{d} \coloneqq (d+1)(b_{0}+1)$.
This upper bound of the bandwidth follows from the observation that $H$ is a $(d+1)$-sketch of $G$ and by Lemma~\ref{lem:bw}.
\end{proof}

Similarly, we can show the W[1]-hardness of $\RISR{d}$ under $\TJ$ and $\TS$ parameterized by the natural parameter $|\Uini|$,
which is often called as the \emph{solution size}.
For $d=1$, the W[1]-hardness is known under both $\TJ$~\cite{ItoKOSUY20} and $\TS$~\cite{BartierBDLM21}.
For an instance of $\RISR{0}$ under $\TS$ parameterized by the solution size $k$,
we apply the reduction in the proof of Theorem~\ref{thm:RISR-chordal}.
The obtained equivalent instance of $\RISR{d}$ under $\TJ$ and $\TS$ has solution size $k (d+1)$,
and thus the following holds.
\begin{corollary}
For every constant $d \ge 0$,
$\RISR{d}$ is W[1]-hard parameterized by the solution size under $\TJ$ and $\TS$.
\end{corollary}

\subsection{Bipartite graphs}
As mentioned in Introduction, on bipartite graphs, 
$\RISR{0}$ is PSPACE-complete under $\TS$ 
but NP-complete under $\TJ$~\cite{LokshtanovM19}.
That is, $\RISR{0}$ is intractable under both rules but in different senses.
In this section, we study the complexity of $\RISR{d}$ on bipartite graphs for $d \ge 1$
and show a kind of reverse phenomenon: the problem becomes trivial with $\TS$ but harder with $\TJ$.

First we observe the triviality under $\TS$.
Observe that if $\setU \onestep \setU'$ under $\TS$,
then there exist adjacent vertices $u \in \setU \setminus \setU'$ and $v \in \setU' \setminus \setU$ such that
$u$ and $v$ have a common neighbor $w \in \setU \cap \setU'$ as $d \ge 1$.
Thus the graph contains a triangle formed by $u, v, w$.
Therefore, in triangle-free graphs (and thus in bipartite graphs), no nontrivial $\TS$-sequence exists.
\begin{observation}
\label{obs:RISR-bigraph-TS}
For every $d \ge 1$, $\RISR{d}$ under $\TS$ on triangle-free graphs is polynomial-time solvable.
\end{observation}

Under the $\TJ$ rule, we first show that the restricted version $\CRISR{d}$ for $d \ge 2$ is already PSPACE-complete on bipartite graphs.
It is known that $\CRISR{2}$ is PSPACE-complete on general graphs~\cite{HanakaIMMNSSV20}.
Our proof for $d=2$ basically follows their reduction from \textsc{Shortest Path Reconfiguration} (SPR, for short)
to $\CRISR{2}$ on general graphs. We start with a reduction for the case of $d = 2$ that has some additional properties
and then increase $d$ by using the properties.

An instance of SPR consists of
a graph $H$ and the vertex sets $P$ and $P'$ of two shortest $x$--$y$ paths in $H$.
SPR asks whether there is a sequence $\langle P_{0}, P_{1}, \dots, P_{q} \rangle$
such that $P_{0} = P$, $P_{q} = P'$,
each $P_{i}$ is the vertex set of a shortest $x$--$y$ path, and 
$|P_{i} \setminus P_{i-1}| = |P_{i-1} \setminus P_{i}| = 1$ for $1 \le i \le q$.
Observe that in polynomial time, we can remove all vertices and edges that are not on any shortest $x$--$y$ path.
Hence, we can assume that the vertex set of $H$ is partitioned into independent sets $D_{1}, D_{2}, \dots, D_{r}$ 
such that $D_{1} = \{x\}$, $D_{r} = \{y\}$, and
each $D_{i}$ is the set of vertices of distance $i-1$ from $x$
(see~\figurename\ref{fig:crisr2-bi}).
%
%
Bonsma~\cite{Bonsma13} showed that SPR is PSPACE-complete.

\begin{theorem}
\label{thm:CRISR-bigraph-TJ}
For every constant $d \ge 2$,
$\CRISR{d}$ is PSPACE-complete on bipartite graphs under $\TJ$.
\end{theorem}
\begin{proof}
Let $\langle H, P, P' \rangle$ be an instance of SPR
with the partition $D_{1}, D_{2}, \dots, D_{r}$ of the vertex set defined as above.
We assume that all sets $D_{i}$ are independent sets.
Let $L$ be the smallest multiple of $2d$ larger than or equal to $\max\{r, 6\}$.
To the graph $H$, we add $L - r$ vertices $v_{r+1}, \dots, v_{L}$
and $L-r+1$ edges $\{y, v_{r+1}\}$, $\{v_{L}, x\}$, and $\{v_{i}, v_{i+1}\}$ for $r+1 \le i \le L-1$.
We set $S = P \cup \{v_{r+1}, \dots, v_{L}\}$ and $S' = P' \cup \{v_{r+1}, \dots, v_{L}\}$.
For $1 \le i \le r$, we set $V_{i} = D_{i}$, and for $r+1 \le i \le L$, we set $V_{i} = \{v_{i}\}$. 
We call the graph constructed $H'$.
See~\figurename~\ref{fig:crisr2-bi}.
\begin{figure}[t]
  \centering
  
  \newcommand{\orgInstance}{
    \foreach \x in {0, 1, 2, 3, 4, 5, 6, 7}{
            \foreach \y in {1, 2, 3, 4}
                {
                    \node at (\x*\xScale,\y*\yScale) (v\x\y) {};
                }
        }
    \draw[thick,black,fill=white] \convexpath{v01,v04}{1.3em};
    \draw[thick,black,fill=white] \convexpath{v71,v74}{1.3em};
    \node[circle, line width=1pt, draw=blue, fill=white] at (0, 2.5*\yScale) (v0) {};
    \node[circle, line width=1pt, draw=blue, fill=white] at (7*\xScale, 2.5*\yScale) (v7) {};
    \node at (v01.north) {\huge $x$};
    \node at (v71.north) {\huge $y$};
    \foreach \u / \v in {
            v0/v11, v0/v12, v0/v13, v0/v14,
            v11/v21, v11/v22, v12/v22, v13/v21, v13/v23, v14/v23, v14/v24,
            v21/v31, v21/v32, v22/v32, v22/v33, v22/v34, v23/v31, v24/v32, v24/v33,
            v31/v41, v32/v41, v32/v43, v32/v44, v33/v41, v33/v43, v34/v42, v34/v43,
            v51/v61, v51/v63, v51/v64, v52/v61, v52/v63, v53/v62, v53/v63, v54/v63, v54/v64,
            v7/v61, v7/v62, v7/v63, v7/v64%
        }{
            \draw (\u) to (\v);
        }
    \draw[thick,black,fill=white] \convexpath{v11,v14}{1.3em};
    \draw[thick,black,fill=white] \convexpath{v21,v24}{1.3em};
    \draw[thick,black,fill=white] \convexpath{v31,v34}{1.3em};
    \draw[thick,black,fill=white] \convexpath{v61,v64}{1.3em};
    \foreach \u in {v11, v22, v32, v44, v51, v62} { \node[circle, fill=\iniClr] at (\u) (c\u) {}; }
    \foreach \u in {v12, v24, v31, v41, v52, v63} { \node[circle, fill=\tarClr] at (\u) (c\u) {}; }
    \foreach \u  [remember=\u as \v (initially v0)] in {v11, v22, v32, v44, v51, v62, v7} { \draw[line width=2pt, \iniClr] (\v) to (\u); }
    \draw[dashed, line width=3pt, white] (v44) to (v51);
    \foreach \u  [remember=\u as \v (initially v0)] in {v12, v24, v31, v41, v52, v63, v7} { \draw[line width=2pt, \tarClr] (\v) to (\u); }
    \draw[dashed, line width=3pt, white] (v41) to (v52);
}

\def\picScale{0.47}
\begin{tikzpicture}[every node/.style={},scale=\picScale, every node/.style={scale=\picScale}]
    \begin{scope}

        \orgInstance
        \huge
        \node at ($(v0)  + (-.8*\xScale, 0)$) (Hlabel) {\Huge $H$};
        \node at ($(v14) + (-1*\xScale , 1.6 *\yScale)$) (D1Label) {$D_1$};
        \node at ($(v14) + (0          , 1.6 *\yScale)$) (D2Label) {$D_2$};
        \node at ($(v24) + (0          , 1.6 *\yScale)$) (D3Label) {$D_3$};
        \node at ($(v34) + (0          , 1.6 *\yScale)$) (D4Label) {$D_4$};
        \node at ($(v64) + (0          , 1.6 *\yScale)$) (Dr1Label) {$D_{r-1}$};
        \node at ($(v64) + (1*\xScale  , 1.6 *\yScale)$) (DrLabel) {$D_{r}$};
    \end{scope}

    \begin{scope}[shift={(0, -6)}]
        \orgInstance
        \foreach \x in {1, 2, 3, 4} {
                \node at ({(7 + \x)*\xScale}, 2.5*\yScale) (w\x) {};
            }

        \foreach \u  [remember=\u as \v (initially v7)] in {w1, w2, w3, w4} {
                \draw[line width=2pt, \tarClr] ($(\v)+(0, 2pt)$) to ($(\u)+(0, 2pt)$);
                \draw[line width=2pt, \iniClr] ($(\v)-(0, 2pt)$) to ($(\u)-(0, 2pt)$);
            }
        \path [side by side={\iniClr}{\tarClr}] (v0) .. controls +(0,4) and +(0,4) ..  (w4);
        \node[circle, line width=1pt, draw=blue, fill=white] at (v0) {};
        \node[circle, line width=1pt, draw=blue, fill=white] at (v7) {};
        \foreach \x in {1, 2,  4} {
                \node[circle, line width=1pt, draw=blue, fill=white] at (w\x) (cw\x) {};
            }
        \draw[line width=5pt, dashed, white] ($(w2)+(1, 0)$) to  ($(w4)+ (-1, 0)$);

        \huge
        \node at ($(v7) + (1.3*\xScale,  1 *\yScale)$) (Vr3Label) {$v_{r+1}$};
        \node at ($(v7) + (2.3*\xScale,  1 *\yScale)$) (Vr4Label) {$v_{r+2}$};
        \node at ($(v7) + (4.3*\xScale,  1 *\yScale)$) (Vr5Label) {$v_{L}$};
        \node at ($(v0)  + (-.8*\xScale, 0)$) (Glabel) {\Huge $H'$};
        \node at ($(v11) + (-1*\xScale, -1.8 *\yScale)$) (V1Label)  {$V_1$};
        \node at ($(v11) + (0,          -1.8 *\yScale)$) (V2Label)  {$V_2$};
        \node at ($(v21) + (0,          -1.8 *\yScale)$) (V3Label)  {$V_3$};
        \node at ($(v31) + (0,          -1.8 *\yScale)$) (V4Label)  {$V_4$};
        \node at ($(v61) + (0,          -1.8 *\yScale)$) (Vr1Label) {$V_{r-1}$};
        \node at ($(v61) + (1*\xScale,  -1.8 *\yScale)$) (Vr2Label) {$V_{r}$};
        \node at ($(v61) + (2*\xScale,  -1.8 *\yScale)$) (Vr3Label) {$V_{r+1}$};
        \node at ($(v61) + (3*\xScale,  -1.8 *\yScale)$) (Vr4Label) {$V_{r+2}$};
        \node at ($(v61) + (5*\xScale,  -1.8 *\yScale)$) (Vr5Label) {$V_{L}$};
    \end{scope}
\end{tikzpicture}
  \caption{Reduction from SPR to $\CRISR{2}$ on bipartite graphs.}
  \label{fig:crisr2-bi}
\end{figure}
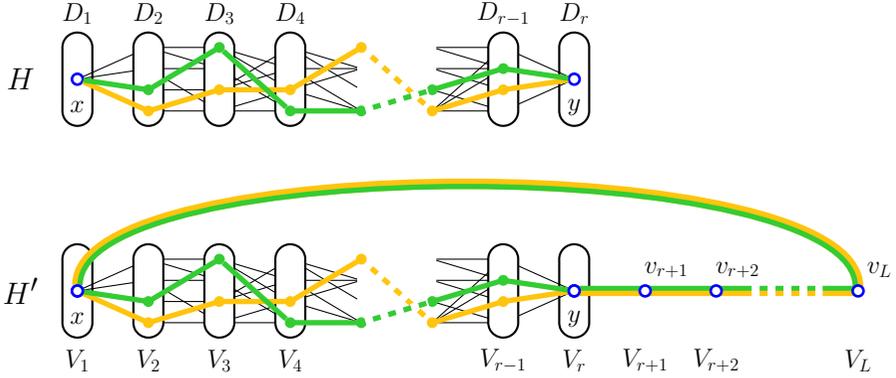

The construction is done if $d = 2$. 
For the cases of $d \ge 3$, we need some additional parts to increase the degree.
For each $i$ with $1 \le i \le L/(2d)$, we attach the gadget in \figurename~\ref{fig:crisrd-bi} to
the consecutive $2d$ sets $V_{2d(i-1)+1}, V_{2d(i-1)+2}, \dots, V_{2di}$.
That is, we add two size-$(d-2)$ independent sets $A_{i}$ and $B_{i}$,
and add all possible edges between $A_{i}$ and $V_{2d(i-1) + (2j-1)}$ for $1 \le j \le d$
and between $B_{i}$ and $V_{2d(i-1) + (2j)}$ for $1 \le j \le d$.
We call the resultant bipartite graph $G$.
Let $I$ denote the set of vertices added, i.e., $I = \bigcup_{1 \le i \le L/(2d)} (A_{i} \cup B_{i})$.
We set $\Uini = S \cup I$ and $\Utar = S' \cup I$.
Clearly, $\Uini$ and $\Utar$ are connected $d$-regular sets of $G$.

We now show that
$\langle G, \Uini, \Utar\rangle$ is a yes-instance of $\CRISR{d}$
if and only if
$\langle H, P, P' \rangle$ is a yes-instance of SPR\@.

To show the if direction, 
assume that there is a reconfiguration sequence $\langle P_{0}, \dots, P_{q} \rangle$ between $P$ and $P'$.
For $0 \le i \le q$, we set $\setU_{i} = P_{i} \cup \{v_{r+1}, \dots, v_{L}\} \cup I$.
Observe that $\setU_{i}$ is a connected $d$-regular set:
$P_{i} \cup \{v_{r+1}, \dots, v_{L}\}$ induces a cycle,
each vertex in the cycle has $d-2$ neighbors in $I$,
and each vertex in $I$ has $d$ neighbors in the cycle.
Since $P_{i} \setminus P_{i-1} = \setU_{i} \setminus \setU_{i-1}$
and $P_{i-1} \setminus P_{i} = \setU_{i-1} \setminus \setU_{i}$ for $1 \le i \le q$,
$\langle \setU_{0}, \dots, \setU_{q} \rangle$ is a $\TJ$-sequence between $\Uini$ and $\Utar$.

To show the only-if direction, 
assume that there is a $\TJ$-sequence $\langle \setU_{0}, \dots, \setU_{q} \rangle$ between $\Uini$ and $\Utar$.
We first show the following fact.
\begin{claim}
$I \subseteq \setU_{i}$ for each $i$ and
$|\setU_{i} \cap V_{j}| = 1$ for each $i$ and $j$.
\end{claim}
\noindent\textit{Proof of Claim. }
The claim is true for $\setU_{0} = \Uini$.
Assume that the claim holds for some $i$ with $0 \le i < q$.
Let $\setU_{i} \setminus \setU_{i+1} = \{u\}$ and $\setU_{i+1} \setminus \setU_{i} = \{v\}$.
Suppose that $u \in I$.
The $d$ neighbors of $u$ in $G[\setU_{i}]$ belong to $d$ different sets $V_{p}, V_{p+2}, \dots, V_{p+2d-2}$ for some $p$.
Since $v$ has the same neighborhood in $G[\setU_{i+1}]$, it has to belong to $I$ as well.
This contradicts the assumption that $v \notin \setU_{i}$ and $I \subseteq \setU_{i}$.
Thus $u$ belongs to some $V_{j}$.
In $G[\setU_{i}]$, $u$ has neighbors in $I$, $V_{j-1}$, and $V_{j+1}$ (where $V_{0} = V_{L}$ and $V_{L+1} = V_{1}$).
Since $L \ge 6$, $v$ has to belong to $V_{j}$. Thus the claim holds.
\hfill
$\triangleleft$

By the claim above, each $\setU_{i}$ includes the vertices $v_{r+1}, \dots, v_{L}$.
Let $P_{i} = \setU_{i} \setminus (I \cup \{v_{r+1}, \dots, v_{L}\})$.
For $2 \le j \le r-1$, the unique vertex in $P_{i} \cap V_{j}$ has exactly two neighbors in $H[P_{i}]$;
one in $V_{j-1}$ and the other in $V_{j+1}$. 
That is, $P_{i}$ is a shortest $x$--$y$ path in $H$.
Therefore, $\langle P_{0}, P_{1}, \dots, P_{q} \rangle$ is a reconfiguration sequence from $P$ to $P'$.
\begin{figure}[t]
  \centering
  
  \def\iclrone{white}
\def\iclrtwo{lightgray}
\def\picScale{0.95}
\begin{tikzpicture}[scale=\picScale, every node/.style={scale=\picScale}]
    \foreach \x in {0, 1, 2, 3, 4, 5, 6, 7, 8} {
            \foreach \y in {1, 2, 3, 4, 5, 6} {
                    \node at (\x, \y*\yScale * 1.3 / 6) (v\x\y) {};
                }
        }

    \foreach \y in {1, 2} {
            \pgfmathsetmacro{\yshift}{ifthenelse(equal(\y,1), "-3", "+1.5")}
            \foreach \x in {0, 1, 2, 3, 4, 5} {
                    \node at (3.5 + \x/4, {(\y + \yshift)*\yScale})  (i\x\y) {};
                }
        }

    \draw[line width=4pt, out=180, in=80] (i02) to (v16);
    \draw[line width=4pt, bend right=20] (i32) to (v36);
    \draw[line width=4pt, bend left=10] (i52) to (v66);
    \draw[line width=4pt, bend left=40] (i01) to (v21);
    \draw[line width=4pt, bend left=5] (i31) to (v41);
    \draw[line width=4pt, out = 0, in=270] (i51) to (v71);

    \draw[thick,black, fill=\iclrone] \convexpath{i01,i51}{6pt};
    \draw[thick,black, fill=\iclrtwo] \convexpath{i02,i52}{6pt};
    \foreach \y in {1, 2} {
            \foreach \x in {0, 1, 2, 5} {
                    \node[circle, minimum size=5pt, thick, draw=black, inner sep=0pt, outer sep=0pt] at (i\x\y) {};
                }
        }

    \draw[thick, dotted, fill=black] ($(i22) + (0.2, 0)$) to ($(i52) - (0.2, 0)$);
    \draw[thick, dotted, fill=black] ($(i21) + (0.2, 0)$) to ($(i51) - (0.2, 0)$);

    \foreach \u/\v in {
            v11/v21, v11/v24, v12/v25, v14/v22, v14/v26, v16/v22, v16/v25,
            v21/v31, v21/v33, v21/v35, v22/v33, v23/v34, v23/v36, v24/v34, v24/v36, v25/v36, v26/v35,
            v31/v42, v31/v44, v31/v46, v32/v45, v32/v46, v34/v42, v34/v44, v35/v42, v35/v44, v35/v45, v36/v46,
            v41/v51, v42/v51, v42/v54, v42/v56, v43/v51, v43/v55, v45/v52, v45/v53, v46/v54, v46/v56,
            v51/v62, v51/v64, v52/v65, v54/v62, v54/v66, v56/v62, v56/v64,
            v61/v73, v61/v76, v62/v73, v64/v71, v64/v75, v66/v73, v66/v74,
            v71/v83, v72/v81, v72/v83, v73/v84, v74/v85, v75/v82, v76/v85,
            v01/v12, v01/v15, v02/v14, v03/v13, v04/v15, v06/v11, v06/v16%
        }{
            \draw (\u) to (\v);
        }
    \draw[thick,white,fill=white] \convexpath{v01,v06}{13pt};
    \draw[thick,white,fill=white] \convexpath{v51,v56}{13pt};
    \draw[thick,white,fill=white] \convexpath{v81,v86}{13pt};
    \foreach \x/\i in {1/(i-1)+1, 2/(i-1)+2, 3/(i-1)+3, 4/(i-1)+4, 6/i-1, 7/i}{
            \pgfmathsetmacro{\col}{ifthenelse(or(or(equal(\x,1),equal(\x,3)),equal(\x,6)), "\iclrone", "\iclrtwo")}
            \pgfmathsetmacro{\labelpos}{ifthenelse(or(or(equal(\x,1),equal(\x,3)),equal(\x,6)), "+0.5", "-1.1")}
            \node at ($(v\x1) - (-0.2, \labelpos)$) {$V_{2d\i}$};
            \draw[thick,black,fill=\col] \convexpath{v\x1,v\x6}{5pt};
        }
    \draw [decorate,decoration={brace,amplitude=6pt}]
    ($(i02.north west) + (0, 0.1)$) -- ($(i52.north east) + (0, 0.1)$) node [black,midway, yshift=13pt]
    {$d-2$};
    \draw [decorate,decoration={brace,amplitude=6pt, mirror}]
    ($(i01.south west) - (0, 0.1)$) -- ($(i51.south east) - (0, 0.1)$) node [black,midway, yshift=-13pt]
    {$d-2$};
    \node at ($(i02) + (-.4, .4)$) {$A_i$};
    \node at ($(i01) + (-.4, -.4)$) {$B_i$};

    \def\rxs{8}
    \def\rys{-0.5}
    \small
    \draw[thick]  (0 + \rxs, 0 + \rys) --(0 + \rxs, 2 + \rys) --(4 + \rxs, 2 + \rys) --(4 + \rxs, 0 + \rys) -- cycle;
    \node at (0.5+\rxs, 0.5+\rys) (e11) {}; \node at (1.2+\rxs, 0.5+\rys) (e12) {};
    \node at (0.5+\rxs, 1.5+\rys) (e21) {}; \node at (1.2+\rxs, 1.5+\rys) (e22) {};
    \draw[line width=4pt] ($(e11)!0.5!(e12)$) to ($(e21)!0.5!(e22)$);
    \draw[thick,black, fill=\iclrone] \convexpath{e11, e12}{6pt};
    \draw[thick,black, fill=\iclrtwo] \convexpath{e21, e22}{6pt};

    \node[anchor=west] at (1.5+\rxs, 0.5+\rys) {independent set};
    \node[anchor=west] at (1.5+\rxs, 1  +\rys) {complete join};
    \node[anchor=west] at (1.5+\rxs, 1.5+\rys) {independent set};
\end{tikzpicture}
  \caption{Increasing the degree from $2$ to $d$.}
  \label{fig:crisrd-bi}
\end{figure}
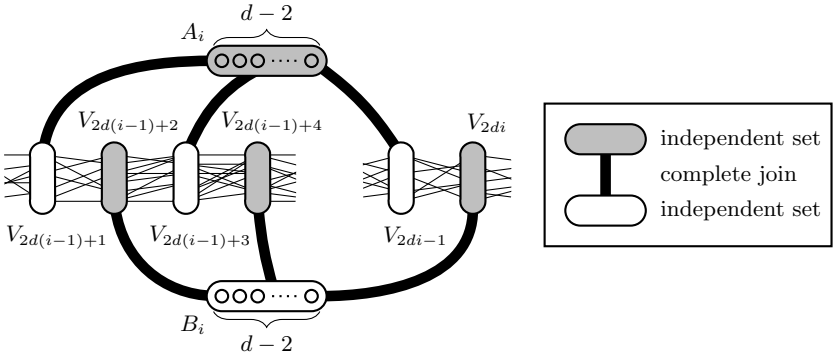
\end{proof}

By slightly modifying the proof of Theorem~\ref{thm:CRISR-bigraph-TJ},
we can show the PSPACE-hardness of $\CRISR{d}$ under $\TS$ on graphs of bounded bandwidth.
\begin{theorem}
\label{thm:CRISR-general}
For every constant $d \ge 2$,
there is a constant $b_{d}$ depending only on $d$ such that 
$\CRISR{d}$ is PSPACE-complete under $\TJ$ and $\TS$ on graphs of bandwidth at most $b_{d}$.
\end{theorem}
\begin{proof}
From an instance $\langle H, P, P' \rangle$ of SPR,
we construct an instance $\langle G, \Uini, \Utar\rangle$ of $\CRISR{d}$ as in Theorem~\ref{thm:CRISR-bigraph-TJ}.
Now we add all possible edges in each $V_{i}$ and call the resultant graph $G'$.
This does not affect the correctness of the arguments in the proof of Theorem~\ref{thm:CRISR-bigraph-TJ}.
At each step, we still have to remove a vertex and add one from the same set $V_{i}$.
Since we add all possible edges in each $V_{i}$, the vertices involved in this step are adjacent.
That is, a $\TJ$-sequence between $\Uini$ and $\Utar$ in $G$
is a $\TS$-sequence between $\Uini$ and $\Utar$ in $G'$, and vice versa.

Now we prove the claim on the bandwidth.
Wrochna~\cite{Wrochna18} showed that there is a constant $b$ such that SPR is PSPACE-complete even if each $D_{i}$ has size at most $b$.
Let us start the reduction in the proof of Theorem~\ref{thm:CRISR-bigraph-TJ} with this restricted version.
Let $W_{i} = A_{i} \cup B_{i} \cup \bigcup_{1 \le j \le 2d} V_{2d(i-1)+j}$ for each $i$.
We can see that $|W_{i}| \le 2db + 2(d-2)$ and that each edge is either in $W_{i}$ for some $i$ or connecting $W_{i}$ and $W_{i+1}$, where $W_{L+1} = W_{1}$.
This implies that a cycle is a $(2db + 2(d-2))$-sketch of $G$, and thus $G$ has bandwidth at most $b_{d} := 2(2db + 2(d-2))$ by Lemma~\ref{lem:bw}.
\end{proof}

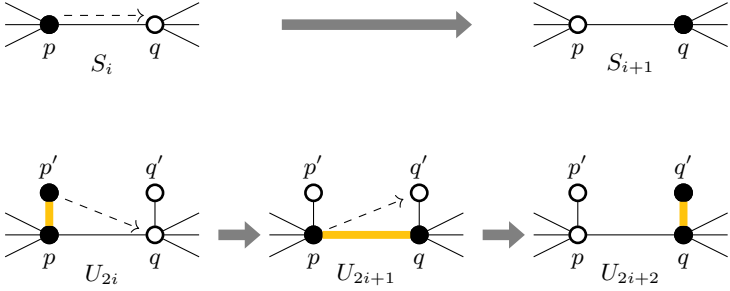
\begin{figure}[htb]
  \centering
  \def\biscale{0.7}
\newcommand{\MatchingBase}{
    \node at (0, 0) (pv1) {};
    \node at (0, 0.5) (pv2) {};
    \node at (0, 1) (pv3) {};

    \node[pq] at (1, 0.5) (p) {};
    \node at ($(p) - (0, 0.5)$) {$p$};
    \node[pq] at (3, 0.5) (q) {};
    \node at ($(q) - (0, 0.5)$) {$q$};

    \node at (4, 0) (qv1) {};
    \node at (4, 0.5) (qv2) {};
    \node at (4, 1) (qv3) {};

    \foreach \u/\v in {
            pv1/p, pv2/p, pv3/p,
            qv1/q, qv2/q, qv3/q,
            p/q%
        } {
            \draw (\u) edge (\v);
        }
}
\newcommand{\MatchingAuxBase}{
    \node at (p) {};
    \node[pq] at ($(p) + (0, 0.8)$) (pp) {};
    \node at ($(pp) + (0, 0.5)$) {$p'$};
    \node[pq] at ($(q) + (0, 0.8)$) (qp) {};
    \node at ($(qp) + (0, 0.5)$) {$q'$};
    \draw (p) edge (pp);
    \draw (q) edge (qp);
}
\begin{tikzpicture}[scale = \biscale,
        matching edge/.style = {draw=\iniClr, line width=3pt},
        pq/.style = {scale=\biscale, circle, line width=1pt, draw=black, fill=white},
        token/.style = {scale=\biscale, circle, line width=1pt, draw=black, fill=black},
    ]
    \def\xshiftAmount{5}
    \def\yshiftAmount{-4}

    \begin{scope}
        \MatchingBase
        \node at ($(p)!0.5!(q) - (0, .75)$) {$S_i$};
        \node[token] at (p) {};
        \draw[dashed, ->] ($(p) + (8pt, 5pt)$) to ($(q) + (-6pt, 5pt)$);
        \draw[-{Triangle[width=10pt,length=6pt]}, line width=4pt, gray]
        ($(qv2) + (1.4, 0)$) to ($(qv2) + (\xshiftAmount, 0)$);
    \end{scope}

    \begin{scope}[shift={(2*\xshiftAmount, 0)}]
        \MatchingBase
        \node at ($(p)!0.5!(q) - (0, .75)$) {$S_{i+1}$};
        \node[token] at (q) {};
    \end{scope}

    \begin{scope}[shift={(0*\xshiftAmount, \yshiftAmount)}]
        \MatchingBase
        \node at ($(p)!0.5!(q) - (0, .75)$) {$U_{2i}$};
        \MatchingAuxBase
        \node[token] at (p) {};
        \node[token] at (pp) {};
        \draw[matching edge] (p) to (pp);
        \draw[dashed, ->, shorten >= 3pt, shorten <= 3pt] (pp) edge (q);

        \draw[-{Triangle[width=10pt,length=6pt]}, line width=4pt, gray] ($(qv2) + (0.2, 0)$) to ($(qv2) + (1, 0)$);
    \end{scope}
    \begin{scope}[shift={(1*\xshiftAmount, \yshiftAmount)}]
        \MatchingBase
        \node at ($(p)!0.5!(q) - (0, .75)$) {$U_{2i+1}$};
        \MatchingAuxBase
        \node[token] at (p) {};
        \node[token] at (q) {};
        \draw[matching edge] (p) to (q);
        \draw[dashed, ->, shorten >= 3pt, shorten <= 3pt] (p) edge (qp);
        \draw[-{Triangle[width=10pt,length=6pt]}, line width=4pt, gray] ($(qv2) + (0.2, 0)$) to ($(qv2) + (1, 0)$);
    \end{scope}
    \begin{scope}[shift={(2*\xshiftAmount, \yshiftAmount)}]
        \MatchingBase
        \node at ($(p)!0.5!(q) - (0, .75)$) {$U_{2i+2}$};
        \MatchingAuxBase
        \node[token] at (q) {};
        \node[token] at (qp) {};
        \draw[matching edge] (q) to (qp);
    \end{scope}

\end{tikzpicture}
  \caption{Simulating a token sliding by two token jumps.}
  \label{fig:bimatching}
\end{figure}

Now we turn our attention back to bipartite graphs
and complete the map of complexity of $\RISR{d}$.
\begin{theorem}
\label{them:RISR-bigraph-TJ}
For every constant $d \ge 1$,
$\RISR{d}$ is PSPACE-complete on bipartite graphs under $\TJ$.
\end{theorem}
\begin{proof}
By Theorem~\ref{thm:CRISR-bigraph-TJ}, it suffices to show that 
$\RISR{1}$ on bipartite graphs is PSPACE-complete under $\TJ$.

We prove by a reduction from $\RISR{0}$ on bipartite graphs under $\TS$, which is PSPACE-complete~\cite{LokshtanovM19}.
Let $\langle H, S, T \rangle$ be an instance of $\RISR{0}$ under $\TS$, where $H$ is bipartite.
We obtain a graph $G$ from $H$ by attaching a pendant vertex $v'$ to each vertex $v$ of $H$.
Formally, we set $V(G) = V(H) \cup \{v' \mid v \in V(H)\}$ and $E(G) = E(H) \cup \{\{v,v'\} \mid v \in V(H)\}$.
We set $\Uini = S \cup \{v' \mid v \in S\}$ and $\Utar = T \cup \{v' \mid v \in T\}$.
This completes the construction. Note that $\Uini$ and $\Utar$ are $1$-regular sets of $G$.
We show that 
$\langle G, \Uini, \Utar \rangle$ is a yes-instance of $\RISR{1}$ under $\TJ$
if and only if
$\langle H, S, T \rangle$ is a yes-instance of $\RISR{0}$ under $\TS$.

To show the if direction,
assume that $\langle S_{0}, \dots, S_{\ell} \rangle$ is a $\TS$-sequence from $S$ to $T$.
For each $i$ with $0 \le i < \ell$, we set 
$\setU_{2i} = S_{i} \cup \{v' \mid v \in S_{i}\}$ and 
$\setU_{2i+1} = \setU_{2i} \setminus \{p'\} \cup \{q\}$,
where $S_{i} \setminus S_{i+1} = \{p\}$
and $S_{i+1} \setminus S_{i} = \{q\}$.
That is, when $S_{i+1}$ is obtained from $S_{i}$ by sliding a token from $p$ to a neighbor $q$,
we obtain $\setU_{2i+2}$ from $\setU_{2i}$ by first jumping the token from $p'$ to $q$ (we obtain $\setU_{2i+1}$),
and then the other token from $p$ to $q'$. See \figurename~\ref{fig:bimatching}.
Since $p,q$ are adjacent in $H$, the sets $\setU_{2i}$, $\setU_{2i+1}$, and $\setU_{2i+2}$ are $1$-regular in $G$.
Therefore, $\langle \setU_{0}, \dots, \setU_{2\ell} \rangle$ is a $\TJ$-sequence from $\setU_{0} = \Uini$ and $\setU_{2\ell} = \Utar$.

To show the only-if direction,
assume that $\langle \setU_{0}, \dots, \setU_{\ell} \rangle$ is a $\TJ$-sequence from $\Uini$ to $\Utar$.
Let $(A,B)$ be a partition of $V(H)$ into independent sets of $H$.
From each $\setU_{i}$, we define
\[
  R_{i} = \{v \in V(H) \mid \{v, v'\} \subseteq \setU_{i}\} \cup \{u \in A \mid \{u,v\} \subseteq \setU_{i}, \; \{u,v\} \in E(H)\}.
\]
That is, $R_{i}$ is obtained by projecting $\setU_{i}$ onto $V(H)$
and then further replacing two adjacent vertices in $\setU_{i} \cap V(H)$ with the one in $A$.
Clearly, each $R_{i}$ is an independent set of $H$ with size $|\setU_{i}|/2 = |S|$. In particular, $R_{0} = S$ and $R_{\ell} = T$.
Hence, it suffices to show that there is a $\TS$-sequence from $R_{i}$ to $R_{i+1}$ for $0 \le i < \ell$.
By the definition of $R_{i}$, the size of the symmetric difference of $R_{i}$ and $R_{i+1}$ is at most two.
Assume that $R_{i} \ne R_{i+1}$ and that $u \in R_{i} \setminus R_{i+1}$ and $v \in R_{i+1} \setminus R_{i}$ are not adjacent,
since otherwise we are done.
Now there must be a common neighbor $w \notin R_{i} \cup R_{i+1}$ of $u$ and $v$
such that $R_{i} \setminus \{u\} \cup \{w\} = R_{i+1} \setminus \{v\} \cup \{w\}$ is an independent set.
That is, $\langle R_{i}, R_{i} \setminus \{u\} \cup \{w\}, R_{i+1} \rangle$ is a $\TS$-sequence.
\end{proof}

\section{Conclusion}
In this paper, we have investigated the computational complexity of \textsc{$d$-Regular Induced Subgraph Reconfiguration} ($\RISR{d}$) and \textsc{Connected $d$-Regular Induced Subgraph Reconfiguration} ($\CRISR{d}$).
We have shown that $\RISR{d}$ is PSPACE-complete for any fixed $d \ge 1$ even on chordal graphs and on bipartite graphs under two well-studied reconfiguration rules, Token Jumping and Token Sliding, except for some trivial cases.
The results give interesting contrasts to known results for $d = 0$, namely \textsc{Independent Set Reconfiguration}.
On chordal graphs, the two reconfiguration rules do not make a difference in the complexity of $\RISR{d}$ for $d \ge 1$, whereas they do make a significant difference for $d = 0$, and on bipartite graphs, the complexity for $d \ge 1$ shows a reverse phenomenon from $d = 0$.
For any fixed $d \ge 2$, $\CRISR{d}$ on bipartite graphs is PSPACE-complete under Token Sliding rule and polynomial-time solvable under Token Jumping rule.

%
%
%
\bibliography{main}
\bibliographystyle{splncs04}

\end{document}